\documentclass[dvips,11pt,letter]{article}

\skip\footins=4ex                                % Space above first footnote.
\usepackage{amsmath}
\usepackage{amsthm}
\usepackage{amssymb}

\usepackage{graphicx}

\usepackage{algorithmic} %seems incompatible with JACM
\usepackage{algorithm} %must be after hyperref import for alg refs to work right

\newcommand{\oneone}{1\!\!1}
\newcommand{\one}[1]{\oneone\left(#1\right)}

\newcommand{\piecewise}[1]{\left\{\begin{array}{@{\;}ll}#1\end{array}\right.}
\newcommand{\set}[1]{\{#1\}}                        % Set (as in \set{1,2,3})
\newcommand{\setof}[2]{\{\,{#1}\::\:{#2}\,\}}        % Set (as in \setof{x}{x > 0})

\newcommand{\sm}{\setminus} % "set minus"
\DeclareMathOperator{\argmin}{argmin}

\newcommand{\R}{\mathbb{R}}                     % Reals.
\newcommand{\Z}{\mathbb{Z}}                     % Integers.

\newcommand{\E}[1]{\mathbf{E}\left[#1\right]}

\newtheorem{theorem}{Theorem} %[section]

\newtheorem{anonCorollary}{Corollary}
\newtheorem{corollary}[theorem]{Corollary}

\newtheorem{lemma}[theorem]{Lemma}

\newcommand{\footcomment}[1]{} % uncomment this line to hide author comments
\newcommand{\margincomment}[1]{} % uncomment this line to hide author comments

\newcommand{\unaryOrdering}{\!\!\mapsto\!\!}
\newcommand{\addApprox}{\textsc{AddApprox}}

\newcommand{\betTour}{\textsc{BetweennessTour}}
\newcommand{\fast}{\textsc{FAST}}

\newcommand{\posns}{\mathcal{P}}

\newcommand{\bp}{\,\,\rule[-0.22em]{0.08em}{1em}\,\,} %binary piecewise, i.e. f \bp g is f on f's domain and g on g's domain
\newcommand{\restrictO}[2]{{#1}_{#2}}

\newcommand{\roundOpt}{\sigma^{\Box}} %rounded pi*

\newcommand{\costlyB}{2 \binom{k}{2} \epsilon \binom{n-1}{k-1}} %

\newcommand{\versions}[2]{#2} %full version
\newcommand{\fullOnly}[1]{\versions{}{#1}}

\begin{document}

\title{\huge
  \textbf{
    Approximation Schemes for the~Betweenness~Problem in Tournaments   
    and Related Ranking Problems\\[2ex]}}
\author{
  Marek Karpinski\thanks{
    Dept. of Computer Science, University of Bonn.
    Parts of this work done while visiting Microsoft Research.
    Email: {\tt marek@cs.uni-bonn.de}}
  \and
  Warren Schudy\thanks{
    Dept. of Computer Science, Brown University.
    Parts of this work done while visiting University of Bonn.
    Email: {\tt ws@cs.brown.edu}}\\[2ex]}
\date{}
\maketitle

%%%%%%
\begin{abstract}
{\normalsize
We design the first \emph{polynomial time approximation schemes (PTASs)} for the
\emph{Minimum Betweenness} problem in tournaments and some related higher arity
ranking problems. This settles the approximation status of the Betweenness problem
in tournaments along with other ranking problems which were open for some time now. The results depend on a new technique of dealing with fragile ranking constraints and could be of independent interest.
}
\end{abstract}

%\setcounter{page}{0}
%\clearpage

%%%%%%
\section{Introduction}

We study the approximability of the Minimum Betweenness problem in tournaments
(see \cite{AA07}) that resisted so far efforts of designing polynomial time approximation algorithms with a constant approximation ratio.
For the status of the general Betweenness problem, see e.g. \cite{O79,CS98,AA07,CGM09}.

In this paper we design the first \emph{polynomial time approximation scheme} (PTAS) for that problem, and generalize it to much more general class of ranking CSP problems, called here \emph{fragile} problems.
To our knowledge it is the first nontrivial approximation algorithm for the Betweenness problem in tournaments.

In the Betweenness problem we are given a ground set of \emph{vertices} and a set of \emph{betweenness constraints} involving $3$ vertices and a \emph{designated} vertex among them.
The cost of a ranking of the elements is the number of betweenness constraints with the designated vertex not between the other two vertices.
The goal is to find a ranking \emph{minimizing} this cost.
We refer to the Betweenness problem in tournaments, that is in instances with a constraint for every triple of vertices, as the \betTour{} or \emph{fully dense} Betweenness problem (see \cite{AA07}).
We consider also the $k$-ary extension $k$-\fast{} of the Feedback Arc Set in tournaments (\fast) problem (see \cite{mathieu09fast,A07,ACN08}).

We extend the above problems by introducing a more general class of \emph{fragile ranking} $k$-CSP problems.
A \emph{constraint} $S$ of a ranking $k$-CSP problem is called \emph{fragile} if no two rankings of the vertices $S$ that both satisfy the constraint differ by the position of a single vertex.
A \emph{ranking} $k$-CSP problem is called \emph{fragile} if all its
constraints are fragile.

We now formulate  our main results.
\begin{theorem} \label{thm:main}
There exists a PTAS for the \betTour{} problem.
\end{theorem}
The above answers an open problem of \cite{AA07} on the approximation status of the Betweenness problem in tournaments.

We now formulate our first generalization.
\begin{theorem} \label{thm:strongFragile}
There exist PTASs for all \emph{fragile ranking} $k$-CSP problems in tournaments.
\end{theorem}

Theorem \ref{thm:strongFragile} entails, among other things, existence of a PTAS for the $k$-ary
extension of \fast{}. A PTAS for 2-\fast{} was given in \cite{mathieu09fast}.

\begin{anonCorollary}
There exists a PTAS for the $k$-\fast{} problem.
\end{anonCorollary}

We  generalize \betTour{} to arities $k \ge 4$ by specifying for each constraint $S$ a pair of vertices in $S$ that must be placed at the ends of the ranking induced by the vertices in $S$. Such constraints do not satisfy our definition of fragile, but do satisfy a weaker notion that we call \emph{weak fragility}. The definition of weakly fragile is identical to the definition for fragile except that only four particular single vertex moves are considered, namely swapping the first two vertices, swapping the last two, and moving the first or last vertex to the other end.
We now formulate our most general theorem.
\begin{theorem} \label{thm:weakFragile}
There exist PTASs for all \emph{weak-fragile ranking} $k$-CSP problems in tournaments.
\end{theorem}

\begin{anonCorollary}
There exists a PTAS for the $k$-\betTour{} problem.
\end{anonCorollary}

Additionally our algorithms are guaranteed not only to find a \emph{low-cost} ranking but also a ranking that is \emph{close to an optimal ranking} in the sense of Kendall-Tau distance. Karpinski and Schudy \cite{karpinski10exact} recently utilized this extra feature to find an improved parameterized algorithm for \betTour{} with runtime $2^{O(\sqrt{OPT/n})} + n^{O(1)}$.

\begin{theorem}
The PTASs of Theorem \ref{thm:weakFragile} additionally return a set of $2^{\tilde O(1/\epsilon)}$ rankings, one of which is guaranteed to be both cheap (cost at most $(1+O(\epsilon))OPT$) and close to an optimal ranking (Kendall-Tau distance $O\left(\frac{poly\left(\frac{1}{\epsilon}\right)OPT}{n^{k-2}}\right)$).
\end{theorem}

All our PTASs are randomized but one can easily derandomize them by exhaustively considering every possible random choice.

\medskip

Section \ref{sec:notation} introduces notations and the problems we study. Section \ref{sec:algo} introduces our algorithm and an intuitive sense of why it works. Section \ref{sec:runtime} analyzes the runtime. The remaining sections analyze the cost of the output of our algorithms.

%%%%%%%%
\section{Notation}\label{sec:notation}
First we state some core notation. Throughout this paper let $V$ refer to the set of $n$ objects (vertices) being ranked and $\epsilon>0$ the desired approximation parameter. Our $O(\cdot)$ hides $k$ but not $\epsilon$ or $n$. Our $\tilde O(\cdot)$ additionally hides $(\log (1/\epsilon))^{O(1)}$. A \emph{ranking} is a bijective mapping from a ground set $S \subseteq V$ to $\set{1,2,3,\ldots,|S|}$. An \emph{ordering} is an injection from $S$ into $\R$. We use $\pi$ and $\sigma$ (plus superscripts) to denote rankings and orderings respectively. Let $\pi^*$ denote an optimal ordering and $OPT$ its cost. We let $\binom{n}{k}$ (for example) denote the standard binomial coefficient and $\binom{V}{k}$ denote the set of subsets of set $V$ of size $k$.

For any ordering $\sigma$ let $Ranking(\sigma)$ denote the ranking naturally associated with $\sigma$. To help prevent ties we relabel the vertices so that $V = \set{1,2,3,\ldots,|V|}$.
We will often choose to place $u$ in one of $O(1/\epsilon)$ positions $\posns(u) = \set{j \epsilon n + u/(n+1), 0 \le j \le 1/\epsilon}$ (the $u/(n+1)$ term breaks ties).
We say that an ordering is a \emph{bucketed ordering} if $\sigma(u) \in \posns(u)$ for all $u$.
Let $Round(\pi)$ denote the bucketed ordering corresponding to $\pi$ (rounding down), i.e.\ $Round(\pi)(u)$ equals $\pi(u)$ rounded down to the nearest multiple of $\epsilon n$, plus $u/(n+1)$.

Let $v \unaryOrdering p$ denote the ordering over $\set{v}$ which maps vertex $v$ to position $p \in \R$.
For set $Q$ of vertices and ordering $\sigma$ with domain including $Q$ let $\restrictO{\sigma}{Q}$ denote the ordering over $Q$ which maps $u \in Q$ to $\sigma(u)$, i.e.\ the restriction of $\sigma$ to $Q$. For orderings $\sigma^1$ and $\sigma^2$ with disjoint domains let  $\sigma^1\bp \sigma^2$ denote the natural combined ordering over $Domain(\sigma^1) \cup Domain(\sigma^2)$. For example of our notations, $\restrictO{\sigma}{Q} \bp v \unaryOrdering p$ denotes the ordering over $Q \cup \set{v}$ that maps $v$ to $p$ and $u \in Q$ to $\sigma(u)$.

A ranking $k$-CSP consists of a ground set $V$ of \emph{vertices}, an arity $k \ge 2$, and a \emph{constraint system} $c$, where $c$ is a function from rankings of $k$ vertices to $\set{0,1}$.\footnote{Our results transparently generalize to the $[0,1]$ case as well, but the 0/1 case allows simpler terminology.}
Note that $c$ depends on the names of the vertices in the domain of its argument. In particular if $u_1$, $u_2$ and $u_3$ are vertices then $c(u_1 \unaryOrdering 1 \bp u_2 \unaryOrdering 2)$,  $c(u_2 \unaryOrdering 1 \bp u_1 \unaryOrdering 2)$ and $c(u_2 \unaryOrdering 1 \bp u_3 \unaryOrdering 2)$ are all different (although $c(u_1 \unaryOrdering 1 \bp u_2 \unaryOrdering 2)$ and $c(u_2 \unaryOrdering 2 \bp u_1 \unaryOrdering 1)$ are the same).
We say that a subset $S \subset V$ of size $k$ is \emph{satisfied} in ordering $\sigma$ of $S$ if $c(Ranking(\sigma)) = 0$. For brevity we henceforth abuse notation and omit the ``$Ranking$'' and  write simply $c(\sigma)$. The objective of a ranking CSP is to find an ordering $\sigma$ (w.l.o.g.\ a ranking) minimizing the number of unsatisfied constraints, which we denote by $C^c(\sigma) = \sum_{S \in \binom{Domain(\sigma)}{k}} c(\restrictO{\sigma}{S})$. We will frequently omit the superscript $c$, in which case it should be understood to be the constraint system of the overall problem we are trying to solve.

Abusing notation we sometimes refer to $S \subseteq V$ as a \emph{constraint}, when we really are referring to $c$ applied to orderings of $S$.
A constraint $S$ is \emph{fragile} if no two orderings that satisfy it differ by the position of a single vertex. In other words constraint $S$ is fragile if $c(\restrictO{\pi}{S}) + c(\restrictO{\pi'}{S}) \ge 1$ for all rankings $\pi$ and $\pi'$ over $S$ that differ by a single vertex move, i.e.\ $\pi' = Ranking(v \unaryOrdering p \bp \restrictO{\pi}{S \sm \set{v}})$ for some $v \in S$ and $p \in (\Z + 1/2)$. An alternate definition is that a satisfied fragile constraint becomes unsatisfied whenever a single vertex is moved, which is why it is called ``fragile.'' Fragility is illustrated in Figure \ref{fig:fragile}.

\begin{figure}[t]
\begin{center}
\includegraphics[width=0.7\textwidth]{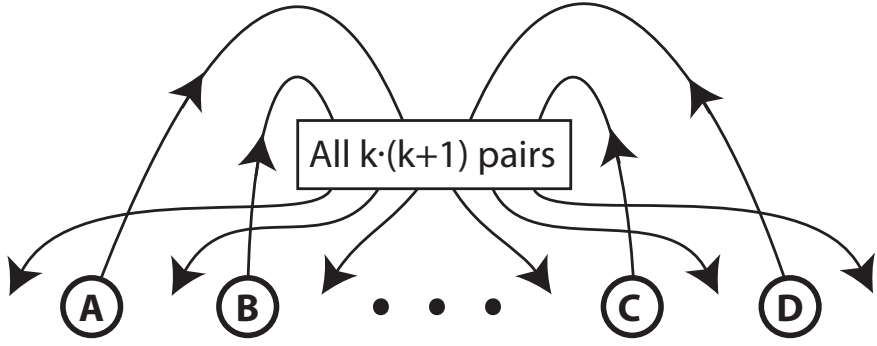}
\end{center}
\caption{An illustration of fragility. For a constraint to be fragile all the illustrated single vertex moves must make any satisfied constraint unsatisfied.}
\label{fig:fragile}
\end{figure}

\begin{figure}[t]
\vspace{1.5em} %so the two figures don't run together
\begin{center}
\includegraphics[width=0.7\textwidth]{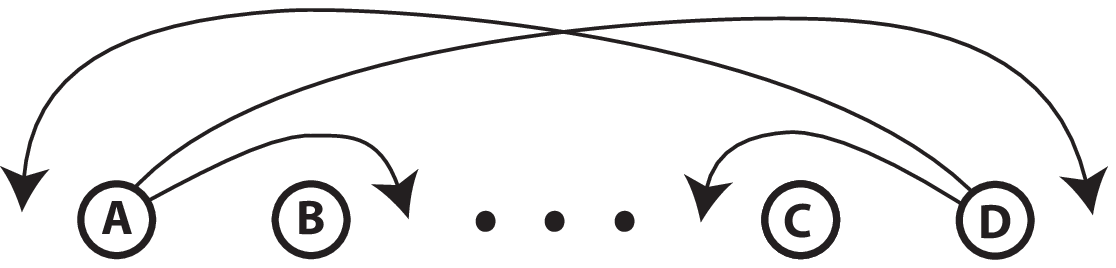}
\end{center}
\caption{An illustration of weak fragility. For a constraint to be weak fragile all the illustrated single vertex moves must make any satisfied constraint unsatisfied.}
\label{fig:weakFragile}
\end{figure}

A constraint $S$ is \emph{weakly fragile} if $c(\restrictO{\pi}{S}) + c(\restrictO{\pi'}{S}) \ge 1$ for all rankings $\pi$ and $\pi'$ that differ by a swap of the first two vertices, a swap of the last two, or a cyclic shift of a single vertex. In other words $\pi' = Ranking(v \unaryOrdering p \bp \restrictO{\pi}{S \sm \set{v}})$ for some $v \in S$ and $p \in \R$ with $(\pi(v),p) \in \set{(1,2 + \frac{1}{2}),(1,k+\frac{1}{2}),(k,k-\frac{3}{2}),(k,\frac{1}{2})}$. Observe that this is equivalent to ordinary fragility for $k \le 3$. Weak fragility is illustrated in Figure \ref{fig:weakFragile}.

Our techniques handle ranking CSPs that are \emph{fully dense} with weakly fragile constraints, i.e.\ every set $S$ of $k$ vertices corresponds to a weakly fragile constraint. Fully dense instances are also known as tournaments.

Let $b^c(\sigma, v, p)$ denote the cost of the constraints involving vertex $v$ in ordering $\restrictO{\sigma}{Domain(\sigma) \sm \set{v}} \bp v \unaryOrdering p$ formed by moving $v$ to position $p$ in ordering $\sigma$. Formally $b^c(\sigma, v, p) = \sum_{Q : \cdots} c(\restrictO{\sigma}{Q} \bp v \unaryOrdering p)$, where the sum is over sets $Q \subseteq Domain(\sigma) \sm \set{v}$ of size $k-1$. Note that this definition is valid regardless of whether or not $v$ is in $Domain(\sigma)$. The only requirement is that the range of $\sigma$ excluding $\sigma(v)$ must not contain $p$. This ensures that the argument to $c(\cdot)$ is an ordering (injective). We will usually omit the superscript $c$ (as with $C$).

We call a non-negative weight function $\set{w_{uv}}_{u,v \in U}$ over the edges of the complete graph induced by some vertex set $U$ a \emph{feedback arc set (FAS) instance}. We can express the feedback arc set problem in our framework by the correspondence $c(u \unaryOrdering x \bp v \unaryOrdering y)=\piecewise{w_{vu} & \text{if } x< y\\ w_{uv} & \text{otherwise}}$.
Abusing notation slightly we also write $C^w(\sigma)$ for $C^c(\sigma)$ with the above $c$.
More concretely $C^w(\sigma) = \sum_{u,v:\sigma(u)>\sigma(v)} w_{uv}$.
Similarly we write $b^{w}(\sigma, v, p) = \sum_{u \ne v} \piecewise{ w_{uv} & \text{if }\sigma(u) > p \\ w_{vu} & \text{if }\sigma(u) < p}$.
If a FAS instance satisfies $\alpha \le w_{uv}+w_{vu} \le \beta$ for all $u,v$ and some $\alpha,\beta>0$ we call it a (weighted) feedback arc set tournament (\fast{}) instance.
It is easy to see that \fast{} captures all possible fragile constraints with $k=2$.
We generalize to $k$-FAST as follows: a $k$-FAST constraint over $S$ is satisfied by one particular ranking of $S$ and no others. Clearly $k$-FAST constraints are fragile.

We generalize \betTour{} to $k \ge 4$ as follows. Each constraint $S$ designates two vertices $\set{u,v}$, which must be the first and last positions, i.e.\ if $\pi$ is the ranking of the vertices in $S$ then $c(\pi) = \one{ \set{\pi(u),\pi(v)} \ne \set{1,k}}$. It is easy to see that \betTour{} constraints are weakly fragile.

We use the following two results from the literature.

\begin{theorem}[\cite{mathieu09fast}] \label{thm:fastPTAS}
Let $w$ be a FAS instance satisfying $\alpha \le w_{uv} + w_{vu} \le \beta$ for $\alpha,\beta>0$ and $\beta/\alpha=O(1)$. 
There is a PTAS for the problem of finding a ranking $\pi$ minimizing $C^w(\pi)$ with runtime $n^{O(1)} 2^{\tilde O(1/\epsilon^6)}$.
\end{theorem}

\begin{theorem}[e.g.\ \cite{AKK95,MS08}] \label{thm:addApproxCSP}
For any $k$-ary MIN-CSP and $\delta > 0$ there is an algorithm that produces a solution with cost at most $\delta n^k$ more than optimal. Its runtime is $n^{O(1)} 2^{O(1/\delta^2)}$.
\end{theorem}

Theorem \ref{thm:addApproxCSP} entails the following corollary.

\begin{corollary}\label{thm:addApproxRank}
For any $\delta > 0$ and constraint system $c$ there is an algorithm \addApprox{} for the problem of finding a ranking $\pi$ with $C(\pi) \le C(\pi^*) + \delta n^k$, where $\pi^*$ is an optimal ranking. Its runtime is $n^{O(1)} 2^{\tilde O(1/\delta^2)}$.
\end{corollary}

%%%%%%%%%%%%%%
\section{Intuition and algorithm}\label{sec:algo}

We are in need for some new techniques different in nature from the techniques used for weighted FAST \cite{mathieu09fast}.

Our first idea in this direction is somehow analogous to the approximation of a differentiable function by a tangent line. Given a ranking $\pi$ and any ranking CSP, the change in cost from switching to a similar ranking $\pi'$ can be well approximated by the change in cost of a particular weighted feedback arc set problem (see proof of Lemma \ref{lem:piThree}). Furthermore if the ranking CSP is fragile and fully dense the corresponding feedback arc set instance is a (weighted) tournament (Lemma \ref{lem:wBarFAST}). So \emph{if} we somehow had access to a ranking similar to the optimum ranking $\pi^*$ we could create this FAST instance and run the existing PTAS for weighted \fast{} \cite{mathieu09fast} to get a good ranking.

We do not have access to $\pi^*$ but we can use a variant of the fragile techniques of \cite{KS09} to get close. We pick a random sample of vertices and guess their location in the optimal ranking to within (an additive) $\epsilon n$. We then create an ordering $\sigma^1$ greedily from the random sample.
We show that this ordering is close to $\pi^*$, in that $|\pi^*(v) - \sigma^1(v)| = O(\epsilon n)$ for all but $O(\epsilon n)$ of the vertices (Lemma \ref{lem:piOne}).

We then do a second greedy step (relative to $\sigma^1$), creating $\sigma^2$. We then identify a set $U$ of \emph{unambiguous} vertices for which we know $|\pi^*(v) - \sigma^2(v)| = O(\epsilon n)$ (Lemma \ref{lem:piTwo}). We temporarily set aside the $O(OPT/(\epsilon n^{k-1}))$ (Lemma \ref{lem:Usmall}) remaining vertices. These two greedy steps are similar in spirit to previous work on ordinary (non-ranking) everywhere-dense fragile CSPs \cite{KS09} but substantially more involved.

We then use $\sigma^2$ to create a (weighted) \fast{} instance $w$ that locally represents the CSP. It would not be so difficult to show that $w$ is a close enough representation for an additive approximation, but we want a multiplicative $1+\epsilon$ approximation. Showing this requires overcoming two obstacles that are our main technical contribution.

Firstly the error in $\sigma^2$ causes the weights of $w$ to have significant error (Lemma \ref{lem:w}) even in the extreme case of $OPT = 0$. At first glance even an exact solution to this \fast{} problem would seem insufficient, for how can solving a problem similar to the desired one lead to a precisely correct solution? Fortunately \fast{} is somewhat special. It is easy to see that a zero-cost instance of \fast{} cannot be modified to change its optimal ranking without modifying an arc weight by at least 1/2. We extend this idea to cases where $OPT$ is small but non-zero (Lemma \ref{lem:piThree}).

The second obstacle is that the incorrect weights in FAST instance $w$ may increase the optimum cost of $w$ far above $OPT$, leaving the PTAS for FAST free to return a poor ranking. To remedy this we create a new FAST instance $\bar w$ by canceling weight on opposing arcs, i.e.\ reducing $w_{uv}$ and $w_{vu}$ by the same amount. The resulting simplified instance $\bar w$ clearly has the same optimum ranking as $w$ but a smaller optimum value. The PTAS for FAST requires that the ratio of the maximum and the minimum of $w_{uv} + w_{vu}$ must be bounded above by a constant so we limit the amount of cancellation to ensure this (Lemma \ref{lem:wBarFAST}). It turns out that this cancellation trick is sufficient to ensure that the PTAS for FAST does not introduce too much error (Lemma \ref{lem:fastCheap}).

Finally we greedily insert the relatively few ambiguous vertices into the ranking output by the PTAS for \fast{} \cite{mathieu09fast} (Appendix \ref{sec:analysis}).

\bigskip

%%%%%%%%%%%%%%%%%%%%%%%%%%%%%%
\begin{algorithm}[btp]
Input: Vertex set $V$, $|V|=n$, arity $k$, system $c$ of fully dense arity $k$ constraints, and approximation parameter $\epsilon>0$.
\begin{algorithmic}[1]
\STATE Run $\addApprox(\epsilon^5 n^k)$ and return the result if its cost is at least $\epsilon^4 n^k$
\STATE Pick sets $T_1,\ldots,T_t$ uniformly at random with replacement from $\binom{V}{k-1}$, where $t=\frac{14 \ln (40/\epsilon)}{\binom{k}{2} \epsilon}$.  Guess (by exhaustion) bucketed ordering $\sigma^0$, which is the restriction of $Round(\pi^*)$ to the sampled vertices $\bigcup_i T_i$, where $\pi^*$ is an optimal ranking.
\STATE Compute bucketed ordering $\sigma^1$ greedily with respect to the random samples and $\sigma^0$: \\
$\sigma^1(u) = \argmin_{p \in \posns(u)} \hat b(u, p)$ where $\hat b(u,p) = \frac{\binom{n}{k-1}}{t} \sum_{i : u \not\in T_i} c(\restrictO{\sigma^0}{T_i} \bp v \unaryOrdering p)$.
\STATE For each vertex $v$: If $b(\sigma^1, v, p) \le 13 k^4 3^{k-1} \epsilon \binom{n-1}{k-1}$ for some $p \in \posns(v)$ then call $v$ \emph{unambiguous} and set $\sigma^2(v)$ to the corresponding $p$ (pick any if multiple $p$ satisfy). Let $U$ denote the set of unambiguous vertices, which is the domain of bucketed ordering $\sigma^2$.
\STATE Compute feedback arc set instance over unambiguous vertices $U$ with weights $\bar w_{uv}^{\sigma^2}$ (see text). Solve it using the \fast{} PTAS \cite{mathieu09fast}. Do single vertex moves until local optimality (with respect to the \fast{} objective function), yielding ranking $\pi^3$ of $U$.
\STATE Create ordering $\sigma^4$ over $V$ defined by $\sigma^4(u) = \piecewise{\pi^3(u) & \text{if }u \in U \\ \argmin_{p=v/(n+1) + j, 0 \le j \le n} b(\pi^3, u, p) & \text{otherwise}}$. In other words insert each vertex $v \in V \sm U$ into $\pi^3(v)$ greedily. 
\STATE Return $\pi^4 = Ranking(\sigma^4)$.
\end{algorithmic}
\caption{A $1+O(\epsilon)$-approximation for weak fragile rank $k$-CSPs in tournaments.\label{alg:main}}
\end{algorithm}
%%%%%%%%%%%%%%%%%%%%%%%%%%%%

For any ordering $\sigma$ with domain $U$ we define a weighted feedback arc set instance $\set{w_{uv}^{\sigma}}_{uv}$ as follows. Let $w_{uv}^{\sigma}$ equal the number of the constraints $\set{u,v} \subseteq S \subseteq U$ with $c(\sigma')=1$ where (1) $\sigma' = (\restrictO{\sigma}{S \sm \set{v}} \bp v \unaryOrdering p)$, (2) $p = \sigma(u)-\delta$ if $\sigma(v) > \sigma(u)$ and $p = \sigma(v)$ otherwise, and (3) $\delta>0$ is sufficiently small to put $p$ adjacent to $\sigma(u)$. In other words if $v$ is after $u$ in $\sigma$ it is placed immediately before $v$ in $\sigma'$. Observe that $0 \le w_{uv} \le \binom{|U|-2}{k-2}$. For any two orderings $\sigma$ and $\sigma'$ we use the abbreviation $C^{\sigma'}(\sigma) = C^{w^{\sigma'}}(\sigma)$.
The following Lemma follows easily from the definitions.

\begin{lemma}\label{lem:objEq}
For any ordering $\sigma$ we have (1) $C^\sigma(\sigma) = \binom{k}{2}C(\sigma)$ and (2) $b^{w^\sigma}(\sigma, v, \sigma(v)) = (k-1) \cdot b(\sigma, v, \sigma(v))$ for all $v$.
\end{lemma}

\begin{proof}
Observe that all $w_{uv}$ that contribute to $C^\sigma(\sigma) $ or $b^{w^\sigma}(\sigma, v, \sigma(v))$ satisfy $\sigma(u)>\sigma(v)$ and hence such $w_{uv}$ are equal to the number of constraints containing $u$ and $v$ that are unsatisfied in $\sigma$. The $\binom{k}{2}$ and $k-1$ factors appear because constraints are counted multiple times.
\end{proof}

For any ordering $\sigma$ we define another weighted feedback arc set instance \\
$\bar w_{uv}^{\sigma} = w_{uv}^\sigma - \min(\frac{1}{10 \cdot 3^{k-1}} \binom{|U|-2}{k-2}, w_{uv}^\sigma, w_{vu}^\sigma)$,
 where $U$ is the domain of $\sigma$. For any orderings $\sigma$ and $\sigma'$ let $\bar C^\sigma(\sigma') = C^{\bar w^\sigma}(\sigma')$. Observe that the feedback arc set problems induced by $w$ and $\bar w$ have the same optimal rankings but $\bar w$ has a smaller objective value and hence they are not equivalent for approximation purposes. In other words $C^\sigma(\pi') - C^\sigma(\pi^\circ) = \bar C^\sigma(\pi') - \bar C^\sigma(\pi^\circ)$ for all rankings $\pi'$ and $\pi^\circ$.

For any orderings $\sigma$ and $\sigma'$ with domain $U$, we say that $\set{u,v} \subseteq U$ is a \emph{$\sigma/\sigma'$-inversion} if $\sigma(u) - \sigma(v)$ and $\sigma'(u)-\sigma'(v)$ have different signs. Let $d(\sigma,\sigma')$ denote the number of $\sigma/\sigma'$-inversions (a.k.a.\ Kendall Tau distance).
We say that $v$ does a \emph{left to right $(\sigma,p,\sigma',p')$-crossing} if $\sigma(v) < p$ and $\sigma'(v) >p'$. 
We say that $v$ does a \emph{right to left $(\sigma,p/\sigma',p')$-crossing} if $\sigma(v) > p$ and $\sigma'(v) < p'$. We say that $v$ does a \emph{$(\sigma,p,\sigma',p')$-crossing} if $v$ does a crossing of either sort.
We say that $u$ \emph{$\sigma/\sigma'$-crosses} $p \in \R$ if it does a $(\sigma,p,\sigma',p)$-crossing.

With these notations in hand we now formalize the ideas described in our Algorithm \ref{alg:main}. The non-deterministic ``guess (by exhaustive sampling)'' on line 2 of our algorithm should be implemented in the traditional manner: place the remainder of the algorithm in a loop over possible orderings of the sample, with the overall return value equal to the best of the $\pi^4$ rankings found. Our algorithm can be derandomized by choosing  $T_1,\ldots,T_t$ non-deterministically rather than randomly; see Section \ref{sec:runtime} for details.

If $OPT \ge \epsilon^4 n^k$ then the first line of the algorithm is sufficient for a PTAS so for the remainder of the analysis we assume that $OPT \le \epsilon^4 n^k$. \fullOnly{For most of the analysis we actually need something weaker, namely that $OPT$ is at most some sufficiently small constant times $\epsilon^2 n^k$. We only need the full $OPT \le \epsilon^4 n^k$ in one place in Section \ref{sec:analysis}.}

%%%%%%%%%%%%%%
\section{Runtime analysis} \label{sec:runtime}

By Theorem~\ref{thm:addApproxRank} the additive approximation step takes time $n^{O(1)} 2^{\tilde O(1/\epsilon^{10})}$. There are at most $(1/\epsilon)^{t \cdot (k-1)} = 2^{\tilde O(1/\epsilon)}$ bucketed orderings $\sigma^0$ to try. The PTAS for FAST takes time $n^{O(1)} 2^{\tilde O(1/\epsilon^6)}$ by Theorem~\ref{thm:fastPTAS}. The overall runtime is 
\[
n^{O(1)} 2^{\tilde O(1/\epsilon^{10})} + 2^{\tilde O(1/\epsilon)} \cdot \left(n^{O(1)} + n^{O(1)} 2^{\tilde O(1/\epsilon^6)} \right) = n^{O(1)} 2^{\tilde O(1/\epsilon^{10})}
.\]

Derandomization increases the runtime of the two algorithms that we use as subroutines to $n^{poly(1/\epsilon)}$. There are at most $n^{t \cdot (k-1)} = n^{\tilde O(1/\epsilon)}$ possible sets $T_1,\ldots T_t$ that the derandomized algorithm must consider. Therefore the overall runtime is 
\[(n^{poly(1/\epsilon)} + n^{poly(1/\epsilon)}\cdot 2^{\tilde O(1/\epsilon)} \cdot n^{poly(1/\epsilon)}=n^{poly(1/\epsilon)}
.\]

%%%%%%%%%%%%%%
\section{Analysis of $\sigma^1$} \label{sec:sigmaOne}

Let $\roundOpt = Round(\pi^*)$.
Call vertex $v$ \emph{costly} if $b(\roundOpt, v, \roundOpt(v)) \ge \costlyB$ and \emph{non-costly} otherwise.

\begin{lemma}\label{lem:fewCostly}
The number of costly vertices is at most $\frac{k \cdot OPT}{\epsilon \binom{k}{2} \binom{n-1}{k-1}}$.
\end{lemma}

\begin{proof}%[Proof of Lemma \ref{lem:fewCostly}]
At most an $\epsilon$ fraction of all pairs of vertices are a $\pi^*$/$\roundOpt$-inversion. Therefore by union bound at most an $\epsilon \binom{k}{2}$ fraction of the $\binom{n-1}{k-1}$ possible constraints involving any particular vertex $v$ contain a $\pi^*$/$\roundOpt$-inversion. Therefore for any costly $v$ we have
\[
\costlyB \le b(\roundOpt, v, \roundOpt(v)) \le b(\pi^*, v, \pi^*(v)) + \epsilon \binom{k}{2} \cdot \binom{n-1}{k-1} 
.\]
Rearranging we get
\[
b(\pi^*, v, \pi^*(v)) \ge \costlyB - \epsilon \binom{k}{2} \cdot \binom{n-1}{k-1} = \epsilon \binom{k}{2} \cdot \binom{n-1}{k-1}
.\]

Finally observe that $k C(\pi^*) = \sum_v b(\pi^*, v, \pi^*(v)) \ge (\text{number costly}) \epsilon \binom{k}{2} \binom{n-1}{k-1}$, completing the proof.
\end{proof}

\begin{lemma}\label{lem:fragile} Let $\sigma$ be an ordering of $V$, $|V|=n$, $v\in V$ be a vertex and $p,p' \in \R$.
Let $B$ be the set of vertices (excluding $v$) between $p$ and $p'$ in $\sigma$. 
Then
$b(\sigma, v, p) + b(\sigma, v, p') \ge \frac{|B|}{(n-1)3^{k-1}}\binom{n-1}{k-1}$. 
\end{lemma}

\begin{proof}%[Proof of Lemma \ref{lem:fragile}]
By definition
\begin{equation}
b(\sigma, v, p) + b(\sigma, v, p') = \sum_{Q : \cdots} \left[ c(\restrictO{\sigma}{Q} \bp v \unaryOrdering p)+c(\restrictO{\sigma}{Q} \bp v \unaryOrdering p') \right] \label{eqn:pen}
\end{equation}
where the sum is over sets $Q \subseteq U \sm \set{v}$ of $k-1$ vertices.

We first consider the illustrative special case of betweenness tournament (or more generally fragile problems with arity $k=3$). Betweenness constraints have a special property: the quantity in brackets in (\ref{eqn:pen}) is at least 1 for every $Q$ that has at least one vertex between $p$ and $p'$ in $\pi$. There are at least $|B| (n-2)/2$ such sets, which can easily be lower-bounded by the desired $\frac{|B|}{(n-1)3^{3-1}}\binom{n-1}{3-1}$.

Returning to the general case of weak fragility, observe that the quantity in brackets in (\ref{eqn:pen}) is at least 1 for every $Q$ that either has all $k-1$ vertices between $p$ and $p'$ in $\sigma^2$ or has one vertex between them and the remaining $k-2$ either all before or all after. To lower-bound the number of such $Q$ we consider two cases.

If $|B| \ge |V|/3$ then the number of such $Q$ is at least $\binom{|B|}{k-1} = \frac{|B|}{k-1}\binom{|B|-1}{k-2} \geq \frac{|B|}{2 \cdot (k-1)3^{k-2}}\binom{n-2}{k-2}$ for sufficiently large $n$.

If $|B| < |V|/3$ then either at least $|V|/3$ vertices are before or at least $|V|/3$ vertices are after hence the number of such $Q$ is at least 
$|B|\binom{|V|/3}{k-2} \ge \frac{|B|}{2 \cdot 3^{k-2}}\binom{n-2}{k-2} \ge \frac{|B|}{(k-1) \cdot 3^{k-1}}\binom{n-2}{k-2}$ for sufficiently large $n$.
\end{proof}

For vertex $v$ we say that a position $p \in \posns(v) $ is  \emph{$v$-out of place} if there are at least $6 \binom{k}{2} 3^{k-1} \epsilon n$ vertices between $p$ and $\roundOpt(v)$ in $\roundOpt$. We say vertex $v$ is \emph{out of place} if $\sigma^1(v)$ is $v$-out of place.

\begin{lemma}\label{lem:random}
The number of non-costly out of place vertices is at most $\epsilon n / 2$ with probability at least 9/10.
\end{lemma}

\begin{proof}%[Proof of Lemma \ref{lem:random}]
Focus on some $v \in V$ and $p \in \posns(v)$. From the definition of out-of-place and Lemma \ref{lem:fragile} we have
\[
b(\roundOpt, v, \roundOpt) + b(\roundOpt, v, p) \ge \frac{6 \binom{k}{2} 3^{k-1} \epsilon n}{(n-1)3^{k-1}}\binom{n-1}{k-1} \ge 6 \epsilon \binom{k}{2} \binom{n-1}{k-1}
\]
for any $v$-out of place $p$. Next recall that for non-costly $v$ we have
\begin{equation}
b(\roundOpt, v, \roundOpt(v)) < \costlyB \label{eqn:onCheap}
\end{equation}
hence 
\begin{equation}
b(\roundOpt, v, p) > 4 \binom{k}{2} \epsilon \binom{n-1}{k-1} \label{eqn:offExpensive}
\end{equation}
for any $v$-out of place $p$.

Recall that
\[
\hat b(v,p) = \frac{\binom{n}{k-1}}{t} \sum_{i : v \not\in T_i} c(\restrictO{\sigma^0}{T_i} \bp v \unaryOrdering p)
\]
for any $p$. Each term of the sum is a 0/1 random variable with mean $\mu(p)=\frac{1}{\binom{n}{k-1}} \sum_{Q \in \binom{V}{k-1} : v \not \in Q} c(\restrictO{\sigma^\Box}{Q} \bp v \unaryOrdering p) = \frac{1}{\binom{n}{k-1}} b(\sigma^\Box, v, p)$. Therefore $\E{\hat b(v,p)} = b(\sigma^\Box, v,p)$. We can bound $\mu(\roundOpt(v)) \le \costlyB / \binom{n}{k-1} \equiv M$ using (\ref{eqn:onCheap}). For any $v$-out of place $p$ we can bound $\mu(p) \ge 2M$ by (\ref{eqn:offExpensive}).

We can bound the probability that sum in $\hat b(v,\roundOpt(v))$ is at least $(1+1/3)Mt$ using a Chernoff bound as 
\[
\exp(-(1/3)^2 M t / 3) \le \exp\left(- \frac{1}{9} \cdot \frac{1}{\binom{n}{k-1}} \cdot 2 \binom{k}{2} \epsilon \binom{n-1}{k-1} \cdot \frac{14 \ln (40/\epsilon)}{\binom{k}{2} \epsilon} \cdot \frac{1}{3}\right) \le \epsilon / 40
\]
for sufficiently large $n$. Similarly for any $v$-out of place $p$ we can bound the probability that $\hat b(v,p)$ is at most $(1-1/3)Mt$ by $\exp(-(1/3)^2 M t / 2) \le (\epsilon / 40)^3$.  Therefore by union bound the probability of some $v$-out of place $p$ having $\hat b(v,p)$ too small is at most $\epsilon^2 / 40^3 \le \epsilon / 40$. Clearly $4(1-1/3) \ge 2(1+1/3)$ so each vertex $v$ is out of place with probability at least $\epsilon /20$. A Markov bound completes the proof.
\end{proof}

\begin{lemma}\label{lem:piOne}
With probability at least 9/10 the following are simultaneously true:
\begin{enumerate}
\item The number of out of place vertices is at most $\epsilon n$.
\item The number of vertices $v$ with $|\sigma^1(v) - \roundOpt(v)| > 3 k^2 3^{k-1} \epsilon n$ is at most $\epsilon n$
\item $d(\sigma^1, \roundOpt) \le 6k^2 3^{k-1} \epsilon n^2$
\end{enumerate}
\end{lemma}
\begin{proof}
By Lemma \ref{lem:fewCostly} and the fact $OPT \le \epsilon^4 n^k$ we have at most $\frac{k \cdot OPT}{\binom{k}{2} \epsilon \binom{n-1}{k-1}} \le \epsilon n / 2$ costly vertices for $n$  sufficiently large. Therefore Lemma \ref{lem:random} implies the first part of the Lemma. We finish the proof by showing that whenever the first part holds the second and third parts hold as well.

Observe that there are exactly $\epsilon n$ vertices in $\roundOpt$ between any two consecutive positions in $\posns(v)$. It follows that any vertex with $|\sigma^1(v) - \roundOpt(v)| > 3k^2 3^{k-1} \epsilon n \ge (6 \binom{k}{2} 3^{k-1} + 1)\epsilon n$ must necessarily be $v$-out of place, completing the proof of the second part of the Lemma.

For the final part observe that if $u$ and $v$ are a $\sigma^1$/$\roundOpt$-inversion and not among the $\epsilon n$ out of place vertices then, by definition of out-of-place, there can be at most  $ 2 \cdot 6 \binom{k}{2} 3^{k-1}\epsilon n$ vertices between $\roundOpt(v)$ and $\roundOpt(u)$ in $\roundOpt$. Each $u$ therefore only $24 \binom{k}{2}3^{k-1}\epsilon n$ possibilities for $v$. Therefore $d(\sigma^1, \roundOpt) \le \epsilon n^2 + 24 \binom{k}{2}3^{k-1}\epsilon n \cdot n / 2 \le 6 \epsilon k^2 3^{k-1}n^2$.
\end{proof}

Our remaining analysis is deterministic, conditioned on the event of Lemma \ref{lem:piOne} holding.

%%%%%%%%%%%%%%
\section{Analysis of $\sigma^2$} \label{sec:sigmaTwo}

The following key Lemma shows the sensitivity of $b(\sigma,v,p)$ to its first and third arguments.
\begin{lemma}\label{lem:bChangeFirst}
For any constraint system $c$ with arity $k \ge 2$, orderings $\sigma$ and $\sigma'$ over vertex set $T \subseteq V$, vertex $v \in V$ and $p,p' \in \R$ we have
\begin{align*}
&1.&|b^c(\sigma, v, p) - b^c(\sigma', v, p')|
&\le \binom{n-2}{k-2} (number\ of\ crossings) + \binom{n-3}{k-3} d(\sigma, \sigma')\\
&2.&|b^c(\sigma, v, p) - b^c(\sigma', v, p')|
&\le \binom{n-2}{k-2} \left(|net\ flow| + k \sqrt{d(\sigma, \sigma')}\right) \\
\end{align*}
where $\binom{n-3}{k-3}=0$ if $k=2$, $(net\ flow)$ is $|\setof{v \in T}{\sigma'(v) > p'}| - |\setof{v \in T}{\sigma(v) > p}|$, and $(number\ of\ crossings)$ is the number of $v \in T$ that do a $(\sigma,p,\sigma',p')$-crossing.
\end{lemma}
\begin{proof}
Fix $\sigma$, $\sigma'$, $T$, $v$, $p$ and $p'$. Let $L$ (resp. $R$) denote the vertices in $T$ that do left to right (resp. right to left) $(\sigma,p,\sigma',p')$-crossings. It is easy to see that a constraint $\set{v} \cup Q$, $Q \in \binom{T \sm \set{v}}{k-1}$ contributes identically to $b(\sigma, v, p)$ and $b(\sigma', v, p')$ unless either:
\begin{enumerate}
\item $Q$ and $(L \cup R)$ have non-empty intersection (or)
\item $Q$ contains a $\sigma/\sigma'$-inversion $\set{s,t}$.
\end{enumerate}
The first part of the Lemma follows easily.

Towards proving the second part we first bound $|L|+|R|$. Observe that $|L| = |R| + (net\ flow)$. Assume w.l.o.g.\ that $(net\ flow) \ge 0$. Observe that every pair $v \in L$ and $w \in R$ are a $\sigma/\sigma'$-inversion, hence
$d(\sigma,\sigma') \ge |L|\cdot|R| = (|R| + (net\ flow))|R| \ge |R|^2$. We conclude that $|L| + |R| = 2|R| + (net\ flow) \le 2\sqrt{d(\sigma,\sigma')} + (net\ flow)$. Therefore the number of constraints of the first type is at most $\binom{n-2}{k-2}(2\sqrt{d(\sigma,\sigma')} + (net\ flow))$.

To simplify we bound
\begin{align*}
\binom{n-3}{k-3}d(\sigma,\sigma') &=\binom{n-2}{k-2}\sqrt{d(\sigma,\sigma')}\cdot \frac{k-2}{n-2} \cdot \sqrt{d(\sigma,\sigma')} \\
&\le \binom{n-2}{k-2}\sqrt{d(\sigma,\sigma')} \cdot (k-2)\frac{\sqrt{n(n-1)/2} }{n-2} \le (k-2)\binom{n-2}{k-2}\sqrt{d(\sigma,\sigma')}
\end{align*}
for sufficiently large $n$.
\end{proof}

Observe that the quantity $net\ flow$ in Lemma \ref{lem:bChangeFirst} is zero whenever $p=p'$ and $\sigma$ and $\sigma'$ are both \emph{rankings}. Therefore we have the following useful corollary.

\begin{corollary} \label{lem:fastSVMlandscape}
Let $\pi$ and $\pi'$ be rankings over vertex set $U$ and $w$ a FAST instance over $U$. Then $|b^w(\pi, v, p) - b^w(\pi', v, p)| \le 2 (\max_{r,s} w_{rs}) \sqrt{d(\pi,\pi')}$ for all $v$ and $p \in \R \sm \Z$.
\end{corollary}

\begin{lemma}\label{lem:Usmall}
For $U$ in Algorithm \ref{alg:main} we have
$|V \sm U| \le \frac{k \cdot OPT}{\epsilon \binom{k}{2}\binom{n-1}{k-1}} = O( \frac{n}{\epsilon} \cdot \frac{OPT}{n^k})$.
\end{lemma}
\begin{proof}
Observe that the number of vertices that $\roundOpt$/$\sigma^1$-cross a particular $p$ is at most $2 \cdot 6k^2 3^{k-1} \epsilon n$ by Lemma~\ref{lem:piOne} (first part). Therefore we apply Lemmas \ref{lem:piOne} and \ref{lem:bChangeFirst}, yielding
\begin{equation}
|b(\roundOpt, v, p) - b(\sigma^1, v, p)| \le \binom{n-2}{k-2}12k^2 3^{k-1}\epsilon n + \binom{n-3}{k-3} 6k^2 3^{k-1}\epsilon n^2 \le 12\epsilon k^4 3^{k-1}\binom{n-1}{k-1} \label{eqn:s1bClose}
\end{equation}
for all $v$ and $p$.

Fix a non-costly $v$. By definition of costly $b(\roundOpt, v, \roundOpt(v)) \le \costlyB \le k^4 3^{k-1} \epsilon \binom{n-1}{k-1}$, hence $b(\sigma^1, v, \roundOpt(v)) \le 13 k^4 3^{k-1} \epsilon \binom{n-1}{k-1}$, so $v \in U$.

Finally recall Lemma \ref{lem:fewCostly}.
\end{proof}

We define $\pi^\circledast$ to be the ranking induced by the restriction of $\pi^*$ to $U$, i.e.\ $\pi^\circledast = Ranking(\restrictO{\pi^*}{U})$.

\begin{lemma}\label{lem:piTwo}
All vertices in the unambiguous set $U$ satisfy $|\sigma^2(v) - \pi^\circledast(v)|  = O(\epsilon n)$. 
\end{lemma}
\begin{proof}
The triangle inequality $|\sigma^2(v) - \pi^\circledast| \le |\sigma^2(v) - \pi^*(v)| + |\pi^*(v) -  \pi^\circledast|$ allows us to instead bound the two terms $|\sigma^2(v) - \pi^*(v)|$ and $|\pi^*(v) -  \pi^\circledast|$ separately by $O(\epsilon n)$. We bound $|\sigma^2(v) - \pi^*(v)|$ first.

Since $\pi^*$ is a ranking the number of vertices $|B|$ between $\pi^*(v)$ and $\sigma^2(v)$ in $\pi^*$ is at least $|\pi^*(v) - \sigma^2(v)| - 1$. Therefore we have
\begin{align}
\frac{|\pi^*(v) - \sigma^2(v)| - 1}{(n-1)3^{k-1}} \binom{n-1}{k-1} &\le b(\pi^*, v, \sigma^2(v)) + b(\pi^*, v, \pi^*(v)) && \text{(Lemma~\ref{lem:fragile})} \nonumber \\
& \le 2 b(\pi^*, v, \sigma^2(v)) && \text{(Optimality of $\pi^*$)} \label{eq:technical1}
.\end{align}

We next apply the first part of Lemma \ref{lem:bChangeFirst} to $\pi^*$ and $\roundOpt$, bounding the number of crossings and $d(\pi^*, \roundOpt)$ using the definition $\roundOpt = Round(\pi^*)$, yielding
\begin{align}
b(\pi^*, v, \sigma^2(v)) &\le b(\roundOpt, v, \sigma^2(v)) +O(\epsilon n^{k-1})
. \label{eq:technical2}\end{align}

Next recalling (\ref{eqn:s1bClose}) from the proof of Lemma \ref{lem:Usmall} we have
\begin{align}
b(\roundOpt, v, \sigma^2(v))&\le b(\sigma^1, v, \sigma^2(v)) + O(\epsilon n^{k-1})
. \label{eq:technical3} \end{align}

Combining (\ref{eq:technical1}), (\ref{eq:technical2}) and (\ref{eq:technical3}) we conclude that $|\pi^*(v) - \sigma^2(v)| = O(\epsilon n)$.

Now we prove $|\pi^*(v) -  \pi^\circledast|=O(\epsilon n)$.
Lemma~\ref{lem:Usmall}, the definition of $\pi^\circledast$, and the assumption that $OPT \le \epsilon^4 n^k$ imply that $|\pi^\circledast(v) - \pi^*(v)| \le \frac{k \cdot OPT}{\epsilon \binom{k}{2}\binom{n-1}{k-1}} = O(\epsilon n)$.
\end{proof}

%%%%%%%%%%%%%%%%%%%%%%%%%%%%%%%%%%%%%%%%%%%%%%%%%%%%%%%%%%%%%%%%%%%%%%%
%%%%%%%%%%%%%%%%%%%%%%%%%%%%%%%%%%%%%%%%%%%%%%%%%%%%%%%%%%%%%%%%%%%%%%%
\section{Analysis of $\pi^3$} \label{sec:piThree}

Note that all orderings and costs in this section are over $U$, not $V$. We note that  Lemma~\ref{lem:Usmall} and the assumption that $OPT\le \epsilon^4 n^k$ is small imply that $|U|=n - O(\epsilon^3 n)$.

\begin{lemma}\label{lem:wBarFAST}
$\frac{1}{3^{k-1}}(1 - 2/10) \binom{|U|-2}{k-2} \le \bar w_{uv}^{\sigma^2} + \bar w_{vu}^{\sigma^2} \le 2 \binom{|U|-2}{k-2}$, i.e.\ $\bar w^{\sigma^2}$ is a weighted FAST instance.
\end{lemma}
\begin{proof}
We prove the more interesting lower-bound and leave the straightforward proof of the upper bound to the reader. Fix $u,v \in U$. We consider two cases.

If there are at least $|U|/3$ vertices between $u$ and $v$ in $\sigma^2$ then we note that by weak fragility every constraint $S \supseteq \set{u,v}$ with all vertices in $S$ between $u$ and $v$ in $\sigma^2$ contributes at least $1$ to $w_{uv}+w_{vu}$. Therefore $w_{uv}+w_{vu} \ge \binom{|U|/3}{k-2} \ge \frac{1}{2 \cdot 3^{k-2}}\binom{n-2}{k-2}$ for sufficiently large $n$ and small $\epsilon$.

If there are at most $|U|/3$ vertices between $u$ and $v$ in $\sigma^2$ then consider constraints with all their vertices either all before or all after $u$ and $v$. We note that by weak fragility each such constraint $S \supseteq \set{u,v}$ contributes at least $1$ to $w_{uv}+w_{vu}$. There are clearly either at least $|U|/3$ vertices  before or at least $|U|/3$ vertices after, hence at least $\binom{|U|/3}{k-2} \ge \frac{1}{2 \cdot 3^{k-2}} \binom{n-2}{k-2}$ constraints for sufficiently large $n$ and small $\epsilon$.

We conclude that $w_{uv} + w_{vu} \ge \frac{1}{2 \cdot 3^{k-2}}\binom{n-2}{k-2} \ge \frac{1}{3^{k-1}}\binom{n-2}{k-2}$. The Lemma follows from the definition of $\bar w$.
\end{proof}

\begin{lemma}\label{lem:footInv}
Assume ranking $\pi$ and ordering $\sigma$ satisfy $|\pi(u) - \sigma(u)| =O( \epsilon n)$ for all $u$. 
For any $u,v$, let $N_{uv}$ denote the number of $S \supset \set{u,v}$ such that not all pairs $\set{s,t} \ne \set{u,v}$ are in the same order in $\sigma$ and $\pi$. We have $N_{uv} =O( \epsilon n^{k-2})$.
\end{lemma}
\begin{proof}
Such a pair $\set{s,t}$ must satisfy $|\pi(s) - \pi(t)| =2\cdot O(\epsilon n)$, but few constraints contain such a pair.
\end{proof}

\begin{lemma}\label{lem:w} The following inequalities hold:
\begin{enumerate}
\item $w_{uv}^{\sigma^2} \le w_{uv}^{\pi^\circledast} + O(\epsilon n^{k-2})$ 
\item $\bar w_{uv}^{\sigma^2} \le (1+O(\epsilon)) w_{uv}^{\pi^\circledast}$
\end{enumerate}
\end{lemma}
\begin{proof}
The only constraints $S \supset \set{u,v}$ that contribute differently to the left- and right-hand sides of the first part are those containing a $\set{s,t} \ne \set{u,v}$ that are a $\sigma^2$/$\pi^\circledast$-inversion. By Lemmas \ref{lem:piTwo} and \ref{lem:footInv} we can bound the number of such constraints by $O(\epsilon n^k)$, completing the proof of the first part.

If $ w_{uv}^{\pi^\circledast} \ge \frac{1}{2 \cdot 3^{k-1}}\binom{|U|-2}{k-2}$ the second part follows from the first part and the trivial fact $\bar w \le w$. Otherwise by the first part we have $w_{uv}^{\sigma^2} < 0.6 \frac{1}{3^{k-1}}\binom{|U|-2}{k-2}$. Therefore by Lemma \ref{lem:wBarFAST} $w_{vu}^{\sigma^2} > 0.2\frac{1}{3^{k-1}}\binom{|U|-2}{k-2}$ hence $\bar w_{uv}^{\sigma^2} = w_{uv}^{\sigma^2} - \min(0.1\frac{1}{3^{k-1}} \binom{|U|-2}{k-2}, w_{uv}^{\sigma^2}) = \min( w_{uv}^{\sigma^2} - 0.1 \frac{1}{3^{k-1}} \binom{|U|-2}{k-2}, 0) \le \min(w_{uv}^{\pi^\circledast}, 0) \le w_{uv}^{\pi^\circledast}$ using the first part of the Lemma in the penultimate inequality.
\end{proof}

\begin{lemma}\label{lem:fastCheap} $~$
\begin{enumerate}
\item 
$\bar C^{\sigma^2}(\pi^\circledast) \le (1+O(\epsilon)) \binom{k}{2} C(\pi^\circledast)$
\item
$\bar C^{\sigma^2}(\pi^3) \le (1+O(\epsilon)) \binom{k}{2} C(\pi^\circledast)$
\item
$\bar C^{\sigma^2}(\pi^3) - \bar C^{\sigma^2}(\pi^\circledast) = O(\epsilon C(\pi^\circledast))$
\end{enumerate}
\end{lemma}
\begin{proof}
From the second part of Lemma \ref{lem:w} and Lemma \ref{lem:objEq} we conclude that
\begin{equation*}
\bar C^{\sigma^2}(\pi^\circledast) \le (1+O(\epsilon)) C^{\pi^\circledast}(\pi^\circledast) = (1+O(\epsilon))\binom{k}{2} C(\pi^\circledast)
.\end{equation*}
proving the first part of this Lemma.

The PTAS for FAST guarantees
\begin{equation}
\bar C^{\sigma^2}(\pi^3) \le (1+O(\epsilon)) \bar C^{\sigma^2}(\pi^\circledast)  \label{eqn:fcOne}
,\end{equation}
which combined with the first part  of this Lemma yields the second part.

Finally the first part of Lemma \ref{lem:w} followed by the first part of this Lemma imply
\[
\bar C^{\sigma^2}(\pi^3) - \bar C^{\sigma^2}(\pi^\circledast) \le O(\epsilon) C^{\sigma^2}(\pi^\circledast) \le O(\epsilon C(\pi^\circledast))
,\]
completing the proof of the third part of this Lemma.
\end{proof}

\begin{lemma}\label{lem:Dsmall}
$d(\pi^3,\pi^\circledast)= O(C(\pi^\circledast) / n^{k-2})$
\end{lemma}
\begin{proof}
$\pi^3$ and $\pi^\circledast$ both have cost at most $2 C(\pi^\circledast)$ (Lemma \ref{lem:fastCheap}, first and second parts) for the FAST instance $\bar w^{\sigma^2}$ (Lemma \ref{lem:wBarFAST}).
\end{proof}

\begin{lemma}\label{lem:23close}
We have
$|\pi^3(v) - \pi^\circledast(v)| = O(\epsilon n)$ for all $v \in U$.
\end{lemma}
\begin{proof}
Fix $v \in U$. In this proof we write $w$ (resp. $\bar w$) as a short-hand for $w^{\sigma^2}$ (resp. $\bar w^{\sigma^2}$). Observe that there are at least $(|\pi^3(v) - \pi^\circledast(v)|-1)$ vertices between $\pi^3(v)$ and $\pi^\circledast(v) + 1/2$ in $\pi^3$. Any such vertex $u$ must contribute $w_{uv}$ to one of $b^{\bar w}(\pi^3, v, \pi^\circledast(v)+ 1/2)$ and $b^{\bar w}(\pi^3,v,\pi^3(v))$ and contribute $w_{vu}$ to the other.
By Lemma \ref{lem:wBarFAST} and local optimality of $\pi^3$ we have
\begin{align*}
(|\pi^3(v) - \pi^\circledast(v)|-1)\frac{(1-2/10)}{3^{k-1}} \binom{|U|-2}{k-2} &\le b^{\bar w}(\pi^3, v, \pi^\circledast(v)+ 1/2) + b^{\bar w}(\pi^3,v,\pi^3(v)) \\
& \le 2b^{\bar w}(\pi^3, v, \pi^\circledast(v)+ 1/2)
.\end{align*}

Now apply Corollary \ref{lem:fastSVMlandscape}
\begin{align*}
b^{\bar w}(\pi^3, v, \pi^\circledast(v)+ 1/2) &\le b^{\bar w}(\pi^\circledast, v, \pi^\circledast(v)) + 2 \sqrt{d(\pi^\circledast, \pi^3)} 2 \binom{|U|-2}{k-2}
\end{align*}
and then recall $\sqrt{d(\pi^\circledast, \pi^3)}=O(\epsilon n)$ by Lemma \ref{lem:Dsmall} and the assumption that $OPT$ is small.

Next
\begin{align}
b^{\bar w}(\pi^\circledast, v, \pi^\circledast(v)) 
&\le (1+O(\epsilon)) b^{w^{\pi^{\circledast}}}(\pi^\circledast, v, \pi^\circledast(v)) && \text{(Second part of Lemma \ref{lem:w})} \notag \\
& = (1+O(\epsilon))b(\pi^\circledast, v, \pi^\circledast(v)) && \text{(Lemma \ref{lem:objEq})} \label{eqn:hello}
\end{align}
Finally
\begin{align*}
b(\pi^\circledast, v, \pi^\circledast(v)) &\le b(\sigma^1, v, \sigma^2(v)) + O(n^{k-2}(\epsilon n + \sqrt{\epsilon^2 n^2})) && \text{(Lemmas \ref{lem:bChangeFirst}, \ref{lem:piOne} and \ref{lem:piTwo})} \\
&= O(\epsilon n^{k-1}) && \text{($v \in U$)}
.\end{align*}
which completes the proof of the Lemma.
\end{proof}

\begin{lemma}\label{lem:piThree}
$C(\pi^3) \le (1+O(\epsilon)) C(\pi^\circledast)$.
\end{lemma}
\begin{proof}
First we claim that
\begin{equation}
|(C(\pi^3) - C(\pi^\circledast)) - (C^{\sigma^2}(\pi^3) - C^{\sigma^2}(\pi^\circledast))| \le E_1 \label{eqn:piThree}
,\end{equation}
where $E_1$ is the number of constraints that contain one pair of vertices $u,v$ in different order in $\pi^3$ and $\pi^\circledast$ and another pair $\set{s,t} \ne \set{u,v}$ with relative order in $\pi^3$, $\pi^\circledast$ and $\sigma^2$ not all equal. Indeed constraints ordered identically in $\pi^3$ and $\pi^\circledast$ contribute zero to both sides of (\ref{eqn:piThree}), regardless of $\sigma^2$. Consider some constraint $S$ containing a  $\pi^3$/$\pi^\circledast$-inversion $\set{u,v} \subset S$. If the restrictions of the three orderings to $S$ are identical except possibly for swapping $u,v$ then $S$ contributes equally to both sides of (\ref{eqn:piThree}), proving the claim.

To bound $E_1$ observe that the number of inversions $u,v$ is $d(\pi^3,\pi^\circledast)\equiv D$. For any $u,v$ Lemmas \ref{lem:23close}, \ref{lem:piTwo} and \ref{lem:footInv} allow us to show at most $O(\epsilon n^{k-2})$ constraints containing $\set{u,v}$ contribute to $E_1$, so $E_1 = O(D \epsilon n^{k-2}) = O(\epsilon C(\pi^\circledast))$ (Lemma \ref{lem:Dsmall}).

Finally bound $C^{\sigma^2}(\pi^3) - C^{\sigma^2}(\pi^\circledast) = \bar C^{\sigma^2}(\pi^3) - \bar C^{\sigma^2}(\pi^\circledast) \le O(\epsilon C(\pi^\circledast))$, where the equality follows from the definition of $w$ and the inequality is the third part of Lemma \ref{lem:fastCheap}.
\end{proof}

%%%%%%%%%%%%%%
\section{Analysis of $\pi^4$}  \label{sec:analysis}
In this section we prove Theorem~\ref{thm:weakFragile}:
\begin{equation}
C(\pi^4) \le (1+O(\epsilon)) OPT \label{eqn:pi4}
\end{equation}
and
\begin{equation}
d(\pi^4, \pi^*) = O\left(\frac{OPT}{poly(\epsilon) n^{k-2}}\right)\label{eqn:pi4close}
.\end{equation}

If $OPT > \epsilon^4 n^k$ then, as discussed in Section \ref{sec:algo}, Equation (\ref{eqn:pi4}) follows from the algorithm and the additive error guarantee. Equation (\ref{eqn:pi4close}) is vacuous in this case. It remains to show (\ref{eqn:pi4}) and (\ref{eqn:pi4close}) in the case that that Sections \ref{sec:sigmaOne}-\ref{sec:piThree} dealt with: $OPT \le \epsilon^4 n^k$.

First we prove (\ref{eqn:pi4}). We consider three contributions to these costs separately: constraints with 0, 1, or 2+ vertices in $V \sm U$. 

The contribution of constraints with 0 vertices in $V \sm U$ to the left- and right-hand sides of (\ref{eqn:pi4}) are clearly $C(\pi^3)$ and $C(\pi^\circledast)$ respectively. We showed $C(\pi^3) \le C(\pi^\circledast) + O(\epsilon) C(\pi^\circledast)$ in Lemma~\ref{lem:piThree}.

Second we consider the contribution of constraints with exactly 1 vertex in $V \sm U$. Consider some $v \in V \sm U$.  We want to compare $b(\pi^3, v, \sigma^4(v))$ and $b((\restrictO{\pi^*}{U}), v, \pi^*(v))$. Let $p$ be the half-integer so that $Ranking(v \unaryOrdering p \bp \restrictO{\pi^\circledast}{U}) = Ranking(v \unaryOrdering \pi^*(v) \bp \restrictO{\pi^*}{U})$. The algorithm's greedy choice minimizes $b(\pi^3, v, \sigma^4(v))$ so $b(\pi^3, v, \sigma^4(v)) \le b(\pi^3, v, p)$. Now using Lemmas \ref{lem:bChangeFirst} and \ref{lem:Dsmall} we have $b(\pi^3, v, p) \le b(\pi^\circledast, v, p) + O(\sqrt{d(\pi^3, \pi^\circledast)} n^{k-2}) = b(\pi^\circledast, v, p) + O(\sqrt{OPT/n^{k}}n^{k-1})$.  Note $b(\pi^\circledast, v, p) = b((\restrictO{\pi^*}{U}), v, \pi^*(v))$. Let $\gamma = OPT / n^k$.
We conclude by Lemma \ref{lem:Usmall} that the contribution of constraints with exactly 1 vertex in $V \sm U$ is $O(|V \sm U| \sqrt{OPT / n^{k}} n^{k-1}) = O(\frac{\gamma^{3/2} n^k}{\epsilon}) = O(\epsilon OPT)$.

Finally by Lemma \ref{lem:Usmall} there are at most $|V \sm U|^2 n^{k-2} = O((\frac{\gamma}{\epsilon})^2 n^2 n^{k-2}) =O(\epsilon^2 OPT)$ constraints containing two or more vertices from $V \sm U$.

This ends the proof of (\ref{eqn:pi4}).

Finally we prove (\ref{eqn:pi4close}). By Lemma \ref{lem:Dsmall} we have
\[
d(\pi^3,\pi^\circledast)= O(C(\pi^\circledast) / n^{k-2})
.\]
Finally a pair of vertices can only be counted in $d(\pi^4,\pi^*)$ but not $d(\pi^3,\pi^\circledast)$ if at least one of the vertices is in the ambiguous set $V \sm U$. By Lemma \ref{lem:Usmall}
$|V \sm U| = O( \frac{n}{\epsilon} \cdot \frac{OPT}{n^k})$ so there at most $O(n \cdot \frac{OPT}{\epsilon n^{k-1}}) = O(\frac{OPT}{\epsilon n^{k-2}})$ such pairs.

%%%%%%%%%%%
\section*{Acknowledgements}
We would like to thank Venkat Guruswami, Claire Mathieu, Prasad Raghavendra and Alex Samorodnitsky for interesting remarks and discussions.

%%%%%%%%%%%
\nocite{gutin09,KS09,CS98,KS09,AA07,ALS09,RV07,MS08,FK99,FKK06,BFK03,AKK95,AFKK02,CGM09}   %shows refs

%%%%%%%%%%%%%%%%%%%%%%%%%%%%%%%%%%%% BIB %%%%%%%%%%%%%%%%%%%%%%%%%%%%%%%%%%
\bibliographystyle{abbrv}
\bibliography{bib}  %name of bib file

%\appendix

%%%%%%%%%%%%%%
\end{document}